\documentclass[dvipsnames,11pt]{article}
\usepackage{amsmath,amssymb,amsthm,color,bm,mathrsfs,extarrows,tikz,graphicx,mathtools,enumitem,fancybox,makecell}
\usepackage[square,sort,comma,numbers]{natbib}
\usepackage{subcaption}
\usepackage[ruled]{algorithm2e}
\usepackage[margin=1in]{geometry}
\usepackage{xcolor,bbm}
\usepackage[pagebackref,hypertexnames=false]{hyperref}
\usepackage{comment}
\usepackage{soul}
\usepackage{thm-restate}
\usepackage{ifthen}
\usepackage{mathdots}
\usepackage{enumitem}
\usepackage{float}

\allowdisplaybreaks

\setlist[itemize]{itemsep=0pt}
\setlist[enumerate]{itemsep=0pt}
\hypersetup{colorlinks=true,urlcolor=blue,linkcolor=blue,citecolor=[rgb]{.42,.56,.14},}
\usetikzlibrary{decorations.pathreplacing,decorations.markings,decorations.pathmorphing,decorations.shapes,arrows.meta,positioning,patterns, quantikz2}

\usepackage[capitalise,nameinlink]{cleveref}
\Crefname{lemma}{Lemma}{Lemmas}
\Crefname{fact}{Fact}{Facts}
\Crefname{question}{Question}{Questions}
\Crefname{theorem}{Theorem}{Theorems}
\Crefname{corollary}{Corollary}{Corollaries}
\Crefname{claim}{Claim}{Claims}
\Crefname{example}{Example}{Examples}
\Crefname{problem}{Problem}{Problems}
\Crefname{definition}{Definition}{Definitions}
\Crefname{notation}{Notation}{Notations}
\Crefname{assumption}{Assumption}{Assumptions}
\Crefname{subsection}{Subsection}{Subsections}
\Crefname{section}{Section}{Sections}
\Crefformat{equation}{(#2#1#3)}

\newtheorem{theorem}{Theorem}[section]
\newtheorem*{theorem*}{Theorem}

\newtheorem*{proposition*}{Proposition}

\newtheorem*{property*}{Property}
\newtheorem{lemma}[theorem]{Lemma}
\newtheorem*{lemma*}{Lemma}

\newtheorem*{corollary*}{Corollary}
\newtheorem*{conjecture*}{Conjecture}

\newtheorem*{fact*}{Fact}

\newtheorem*{exercise*}{Exercise}

\newtheorem*{hypothesis*}{Hypothesis}

\theoremstyle{definition}
\newtheorem{definition}[theorem]{Definition}

\newtheorem{exercise-easy}[theorem]{Exercise}
\newtheorem{exercise-med}[theorem]{Exercise}
\newtheorem{exercise-hard}[theorem]{Exercise$^\star$}
\newtheorem{claim}[theorem]{Claim}
\newtheorem*{claim*}{Claim}

\newtheorem{remark}[theorem]{Remark}
\newtheorem*{remark*}{Remark}

\newtheorem*{observation*}{Observation}

\DeclareSymbolFont{extraup}{U}{zavm}{m}{n}
\DeclareMathSymbol{\varheart}{\mathalpha}{extraup}{86}
\DeclareMathSymbol{\vardiamond}{\mathalpha}{extraup}{87}

\DeclareMathOperator*{\E}{\mathbb E}

\renewcommand{\Pr}{\operatorname*{\mathbf{Pr}}}

\newcommand{\eps}{\varepsilon}
\newcommand{\abs}[1]{\left| #1 \right|}
\newcommand{\vabs}[1]{\left\| #1 \right\|}

\newcommand{\pbra}[1]{\left( #1 \right)}
\newcommand{\sbra}[1]{\left[ #1 \right]}
\newcommand{\cbra}[1]{\left\{ #1 \right\}}
\newcommand{\floorbra}[1]{\left\lfloor #1 \right\rfloor}
\newcommand{\ceilbra}[1]{\left\lceil #1 \right\rceil}

\renewcommand{\mid}{\,\middle\vert\,}
 
\newcommand{\Bin}{\mathsf{Bin}}
\newcommand{\bin}{\{0,1\}}

\newcommand{\indicator}{\mathbbm{1}}

\newcommand{\Ecal}{\mathcal{E}}

\newcommand{\Ab}{\mathbf{A}}
\newcommand{\Bb}{\mathbf{B}}
\newcommand{\Cb}{\mathbf{C}}

\newcommand{\Db}{\mathbf{D}}
\newcommand{\Fb}{\mathbf{F}}

\newcommand{\ib}{\mathbf{i}}

\newcommand{\Pb}{\mathbf{P}}

\newcommand{\Qb}{\mathbf{Q}}

\newcommand{\Ub}{\mathbf{U}}
\newcommand{\xb}{\mathbf{x}}
\newcommand{\Xb}{\mathbf{X}}
\newcommand{\Yb}{\mathbf{Y}}
\newcommand{\yb}{\mathbf{y}}
\newcommand{\Zb}{\mathbf{Z}}

\newcommand{\tvdist}[1]{\vabs{#1}_\mathsf{TV}}

\renewcommand{\tilde}{\widetilde}

\title{On the Advantage of Adaptivity for Sampling with Cell Probes}

\author{
Farzan Byramji\thanks{UC San Diego. Email: \href{mailto:fbyramji@ucsd.edu}{\texttt{fbyramji@ucsd.edu}}. Supported by Simons Investigator Award \#929894, and NSF Awards CCF-2425349 and AF: Medium 2212136.}
\and 
Daniel M.~Kane\thanks{UC San Diego. Email: \href{mailto:dakane@ucsd.edu}{\texttt{dakane@ucsd.edu}}. Supported by NSF Medium Award CCF-2107547.}
\and 
Jackson Morris\thanks{UC San Diego. Email: \href{mailto:jrm035@ucsd.edu}{\texttt{jrm035@ucsd.edu}}.}
\and 
Anthony Ostuni\thanks{UC San Diego. Email: \href{mailto:aostuni@ucsd.edu}{\texttt{aostuni@ucsd.edu}}.}
}

\date{}

\begin{document}

\maketitle

\begin{abstract}
    We construct an explicit distribution $\mathbf{D}$ over $\{0,1\}^N$ that exhibits an essentially optimal separation between adaptive and non-adaptive cell-probe sampling.
    The distribution can be sampled exactly when each output bit is allowed two adaptive probes to an arbitrarily long sequence of independent uniform symbols from $[N]$.
    In contrast, any non-adaptive sampler requires $\tilde{\Omega}(N)$ non-adaptive cell probes to generate a distribution with total variation distance less than $1-o(1)$ from $\mathbf{D}$.
    This provides a $2$-vs-$\tilde{\Omega}(N)$ separation for sampling with adaptive versus non-adaptive cell probes, improving upon the $2$-vs-$\tilde{\Omega}(\log N)$ separation of Yu and Zhan (ITCS '24) and the $(\log N)^{O(1)}$-vs-$N^{\Omega(1)}$ separation of Alekseev, G\"o\"os, Myasnikov, Riazanov, and Sokolov (STOC '26).
\end{abstract}

\section{Introduction}
In recent years, the family of computational tasks involving \emph{sampling} has received increased attention, particularly from a complexity-theoretic viewpoint. Broadly, the goal is to take uniform randomness as input and produce samples from some target distribution. Sampling has connections and applications to many other areas of complexity including data structure lower bounds \cite{viola2012complexity, lovett2011bounded, beck2012large, viola2020sampling, chattopadhyay2022space, viola2023new, yu2024sampling, kane2024locality, alekseev2025sampling}, quantum-classical separations \cite{WP26, kane2024locality, GKM+26}, randomness extractors \cite{de2012extractors, viola2014extractors, chattopadhyay2022space, BKMO26}, and more.

One may study the complexity of sampling tasks in various computational models (e.g., circuits, branching programs, bounded-space Turing machines) and with a variety of randomness sources, but in this work we focus on the \emph{cell-probe} model with uniformly random inputs.
In this setting the sampler has query access to arbitrarily many independent samples drawn uniformly from $[N] \coloneqq \{1,2, \dots, N\}$, and outputs are determined by arbitrary functions depending on some subset of inputs queried, i.e., our sampler corresponds to some function $f\colon [N]^m \to \bin^{N}$, and we consider the distribution on $\bin^{N}$ which results from evaluating $f$ on uniformly random input. We write $f(\Ub_{[N]}^m)$ to denote this distribution where $\Ub_{[N]}^m$ is the uniform distribution on $[N]^m$. We say that $f$ is $d$-local if every output bit is a function of at most $d$ input symbols; that is, it can be computed with at most $d$ non-adaptive cell probes.

We also consider an adaptive version of this model wherein the sampler can perform intermediate computations with later queries depending on the results of earlier ones. Indeed, the focus of this work is on understanding the relative power of adaptivity for sampling tasks in this model: 
\begin{center}
    \emph{How much can adaptive queries help to sample in the cell-probe model?}
\end{center}
The challenge of exhibiting a distribution which requires asymptotically fewer adaptive queries than non-adaptive queries was raised by Viola in \cite{viola2020sampling} and reiterated in \cite{viola2023new}. He suggested the following candidate distribution $\Db_V$ on $[N]^N$.
Fix some integer $k = N^{1-\Omega(1)}$, and sample a string $r \in [N]^k$ uniformly at random.
Each output symbol $x_i$ of $\Db_V$ is determined by independently and uniformly sampling an index $j \in [k]$ and setting $x_i = r_j$.
It is easy to see that two adaptive probes are sufficient to sample from this distribution (assuming $k|N$), but a separation was not known until Yu and Zhan \cite{yu2024sampling} proved that $\Omega(\log{N}/\log\log{N})$ non-adaptive probes are necessary to sample from $\Db_V$ (even approximately).

More recently, Alekseev, G\"o\"os, Myasnikov, Riazanov, and Sokolov \cite{alekseev2025sampling} showed that sampling a uniform random permutation in $[N]^N$ also gives a separation between adaptive and non-adaptive cell-probe samplers. They proved that $N^{\Omega(1)}$ non-adaptive probes are required, while Czumaj \cite{Czumaj15} showed that $O(\log^2 N)$ adaptive probes suffice to get polynomially small distance. Alekseev~et~al.~\cite{alekseev2025sampling} also showed that $(\log N)^{\Omega(1)}$ adaptive probes are necessary, and highlighted obtaining an $O(1)$-vs-$N^{\Omega(1)}$ separation as an interesting open question. Weaker lower bounds on the number of non-adaptive probes required to sample permutations had been shown earlier in \cite{viola2020sampling,  yu2024sampling}.

In this work, we give an essentially optimal $O(1)$-vs-$\tilde{\Omega}(N)$ separation.

\begin{theorem}\label{thm:main_cor}
    For any $\eps > 0$ and infinitely many $N$, there exists an explicit distribution $\Db_1$ over $\bin^N$ which can be sampled exactly by two adaptive probes to independent uniform symbols from $[N]$, yet for any $(N^{1-\eps})$-local function $f\colon [N]^m \to \bin^N$, we have
    \begin{equation}\label{eq:opt_error}
        \tvdist{f(\Ub_{[N]}^m) - \Db_1} \ge 1 - \exp(-N^{\Omega_\eps(1)}).
    \end{equation}

    Similarly, there exists a constant $c > 0$ such that for infinitely many $N$, there exists an explicit distribution $\Db_2$ over $\bin^N$ which can be sampled exactly by two adaptive probes to independent uniform symbols from $[N]$, yet for any $(N / \log^c N)$-local function $g\colon [N]^m \to \bin^N$, we have
    \begin{equation}\label{eq:opt_locality}
        \tvdist{g(\Ub_{[N]}^m) - \Db_2} \ge 1 - \exp(-\Omega(\log^2 N)).
    \end{equation}
\end{theorem}

\begin{remark}\label{rmk:optimality}
    These results are essentially the best possible in several ways.
    \begin{itemize}
        \item Any adaptive sampler which only makes one cell probe (per output bit) is trivially equivalent to a non-adaptive sampler, so at least two probes are necessary for a separation.

        \item The error in \Cref{eq:opt_error} cannot be appreciably improved, since any distribution over $\bin^N$ can be approximated to distance $1 - e^{-\Omega(N)}$ by a point mass (which needs no locality to generate).

        \item The locality in \Cref{eq:opt_locality} cannot be appreciably improved, since a $\Theta(N/\log N)$-local sampler can use the probed cells as a shared seed with $2^{\Theta(N)}$ possible values to approximate any distribution over $\bin^N$ with exponentially small error (see, e.g., \cite[Lemma 5.2]{viola2012complexity}).

        \item The alphabet size $s$ of the random symbols must be large for such a separation to hold.
        If $s$ were smaller than any polynomial in $N$, then any sampler making $d = O(1)$ adaptive probes could be simulated by one making $s^{O(d)} \ll N^{\Omega(1)}$ non-adaptive probes by considering all possible queried inputs. 
    \end{itemize}
\end{remark}

The distributions in \Cref{thm:main_cor} are instantiations of the following distribution with different choices of parameters.

\begin{definition}[Distribution $\Db_{n,\ell}$]\label{def:distD}
    Let $n \ge 2$ be a power of two, and let $\ell \ge 1$ be an integer.
    We define $\Db_{n,\ell}$ to be the distribution over $(\ib^{(1)}, \dots, \ib^{(\ell)},\xb,\yb^{(1)}, \dots, \yb^{(\ell)}) \in \bin^{\ell\log(n) + (\ell+1)n}$, where each $\ib^{(j)}$ is independently uniform over $\bin^{\log(n)}$ (viewed as an element of $[n]$), $\xb$ is uniform over $\bin^n$, and $\yb^{(j)}_k = \xb_{k+\ib^{(j)} \bmod n}$ for all $j \in [\ell]$ and $k\in [n]$.
\end{definition}

Note that $\Db_{n,\ell}$ is similar in spirit to the aforementioned distribution $\Db_V$, but with a more explicit indexing structure.
In $\Db_V$, each output samples a private index $j \in [k]$ and outputs $r_j$.
This produces clusters of coordinates which must take the same value, although the output itself does not reveal which coordinates belong to which cluster.
Our method instead explicitly produces many random groups of $\ell+1$ outputs that must correlate and encodes which clusters those must be in the index portion of the output.

As with $\Db_V$, the distribution $\Db_{n,\ell}$ can be sampled with two adaptive cell probes.

\begin{claim}\label{clm:2_query_Dnell}
    Let $n,\ell$ be as in \Cref{def:distD}, and let $N$ be a (positive) integer multiple of $2n$.
    Then $\Db_{n,\ell}$ can be sampled exactly by two adaptive probes to independent uniform symbols from $[N]$.
\end{claim}
\begin{proof}
    Let the first $\ell + n$ random input symbols be $S_1,\dots, S_\ell, R_1,\dots, R_n \in [N]$.
    For each $j \in [\ell]$, use $S_j$ to generate the shift $\ib^{(j)}$ using some balanced mapping from $[N] \to [n]$.
    Similarly for each $k \in [n]$, use $R_k$ to generate the bit $\xb_k$ using some balanced mapping from $[N] \to \bin$.
    (Since $(2n) | N$, both mappings exist.)
    To produce $\yb^{(j)}_k$, first probe $S_j$ to learn the shift $\ib^{(j)}$, then probe $R_{k + \ib^{(j)} \bmod n}$ and output the corresponding bit.
\end{proof}

Our main technical result is the following lower bound, which shows that $\Db_{n,\ell}$ is hard to sample non-adaptively.

\begin{restatable}{theorem}{thmmainlb}\label{thm:main-lb}
    Let $m \ge d \ge 1$, $\ell \ge 1$, and $n, N \ge 2$ be integers where $n$ is a power of two.
    Additionally, let $f\colon [N]^m \to \bin^{\ell\log(n) +(\ell+1)n}$ be a $d$-local function, and let $\Db_{n,\ell}$ be the distribution in \Cref{def:distD}.
    Then $f(\Ub_{[N]}^m)$ and $\Db_{n, \ell}$ have total variation distance at least
    \[
        1 - \exp\pbra{-\Omega\pbra{\frac{n}{\ell^2 d \log N}}} - d\left(\frac{\Theta(d\log N)}{n}\right)^\ell - \exp\big(d\ell \log(n)\log(N) -\Omega(n)\big).
    \]
\end{restatable}

Taking \Cref{thm:main-lb} as a black-box, we can immediately derive \Cref{thm:main_cor}.

\begin{proof}[Proof of \Cref{thm:main_cor} (assuming \Cref{thm:main-lb})]
    Both $\Db_1$ and $\Db_2$ will be chosen to be padded versions of $\Db_{n, \ell}$ (for different values of $\ell$), where $n \ge 2$ is a power of 2, $\ell \ge 1$ is an integer, and $N = 4\ell n$.
    Note that under these conditions, \Cref{clm:2_query_Dnell} guarantees we can exactly sample $\Db_{n, \ell}$ with two adaptive cell probes, so the same must also be true when the distribution is padded with zeros.

    In the case of $\Db_1$, we set $\delta = 3\eps/4$, $\ell = \floorbra{n^{\delta/3}}$, and $d = \ceilbra{n^{1-\delta}}$.
    Then \Cref{thm:main-lb} implies that for any $d$-local $f\colon [N]^m \rightarrow \bin^{\ell \log n + (\ell+1)n}$, we have
    \[
        \tvdist{f(\Ub_{[N]}^m) - \Db_{n, \ell}} \ge 1 - \exp(-\widetilde{\Omega}(n^{\delta/3})) \geq 1 - \exp(-\Omega(N^{\eps/6}))
    \]
    for large enough $n$. Since $d \geq N^{1-\eps}$ and $\Db_{n, \ell}$ is a projection of $\Db_1$, the desired lower bound for $\Db_1$ follows.
    Similarly in the case of $\Db_2$, we set $\ell =  \floorbra{\log^2 n}$ and $d = \ceilbra {n/\log^7 n}$. Then for any $d$-local $g$, we have 
    \[
            \tvdist{g(\Ub_{[N]}^m) - \Db_{n, \ell}} \ge 1 - \exp(-\Omega(\log^2 n)) \geq 1 - \exp(-\Omega(\log^2 N))
        \]
    for large enough $n$. Using $d \geq N/\log^{10} N$, we obtain the desired lower bound for $\Db_2$. 
\end{proof}

We will prove \Cref{thm:main-lb} in \Cref{sec:main_pf} after setting up some preliminary material in the subsequent section.

\section{Preliminaries}

In this section, we review some basic notation, terminology, and standard results.
For a positive integer $n$, we use $[n]$ to denote the set $\cbra{1,2,\dots, n}$.
All logarithms given in the paper are base 2.
The indicator function is denoted by $\indicator(\cdot)$.

\paragraph{Asymptotics.}
We use the standard $\Omega(\cdot), O(\cdot), \Theta(\cdot)$ asymptotic notation to hide universal positive constants.
Occasionally, we will use subscripts to indicate an unspecified dependence on a particular parameter (e.g., $\Omega_\eps(1)$).
Additionally, we write $\tilde{\Omega}(\cdot)$ to suppress polylogarithmic factors.

\paragraph{Cell-Probe Samplers.}

We study sampling in the cell-probe model. 
A sampler is given access to an arbitrarily long sequence of independent random cells, each uniformly distributed over $[N]$, and produces an output in $\bin^n$ according to probes made to the random cells.
The complexity of the sampler is the maximum number of cells queried in order to produce any single output bit.

A central distinction is whether these probes are adaptive. 
In an adaptive sampler, the location of a later probe may depend on the values observed in earlier probes. 
In a non-adaptive sampler, all probe locations for a given output bit are fixed in advance.
In the non-adaptive setting, this cost is often called the \emph{locality} of the sampler. 
Thus, a function $f\colon [N]^m \to \bin^n$ is said to be $d$-local if each output bit of $f$ depends on at most $d$ input cells.

\paragraph{Probability.}
We use bold letters to denote probability distributions and random variables.
We reserve $\Ub_{[n]}$ for the uniform distribution over $[n]$.
For an event $\Ecal$, we define $\Xb(\Ecal)$ to be the probability mass assigned to $\Ecal$ by $\Xb$.
For a function $f$, we use $f(\Xb)$ to denote the output distribution of $f(\xb)$ on randomly drawn $\xb\sim \Xb$.

Given a distribution $\Xb$ and positive integer $t$, we use $\Xb^t$ to denote the $t$-fold product distribution $\Xb \times \cdots \times \Xb$.
If $s$ is a finite set, we write $\Xb^s$ to emphasize that the coordinates of $\Xb^{|s|}$ are indexed by $s$.
We refer to $\Xb$ as a mixture if it can be written as a convex combination of other distributions.
That is, there exist $c_1, \dots, c_k \in [0,1]$ with $\sum_i c_i = 1$ and distributions $\Xb_1, \dots, \Xb_k$ such that $\Xb(\Ecal) = \sum_{i=1}^k c_i \cdot \Xb_i(\Ecal)$ for every event $\Ecal$.
Occasionally, we write this more concisely as $\Xb = \sum_{i=1}^k c_i \Xb_i$.
We measure the similarity of two (discrete) distributions $\Pb$ and $\Qb$ by the \emph{total variation (TV) distance} 
\[
    \tvdist{\Pb - \Qb} = \max_{\text{event } \Ecal} \Pb(\Ecal) - \Qb(\Ecal) = \frac{1}{2}\sum_x \abs{\Pb(x) - \Qb(x)}.
\]
We say $\Pb$ is \emph{$\eps$-close} to $\Qb$ if $\tvdist{\Pb - \Qb} \le \eps$, and \emph{$\eps$-far} otherwise.

\paragraph{Concentration Inequalities.}
We will need the following standard concentration inequality.
\begin{lemma}[Chernoff]\label{lem:chernoff}
    Let $\Xb_1, \dots, \Xb_n$ be independent indicator random variables with $\Pr[\Xb_i=1] \ge p$ for all $i$.
    Then for any $\delta \in (0,1)$, we have
    \[
        \Pr\sbra{\sum_i \Xb_i \le (1-\delta) pn} \le \exp\pbra{-\Omega(\delta^2 pn)}.
    \]
\end{lemma}
We will also require a variant of \Cref{lem:chernoff} in which the random variables are only partially independent.
There are a plethora of standard results giving concentration under these weaker conditions (e.g., McDiarmid's or Azuma's inequalities).
For our purposes, a lesser-known inequality derived from Shearer's lemma (from information theory) by Gavinsky, Lovett, Saks, and Srinivasan \cite{GLSS15} appears to provide the best bounds.
A collection of random variables $\Xb_1, \dots, \Xb_n$ is said to be a \emph{read-$k$ family} if they are functions of independent random variables $\Yb_1, \dots, \Yb_m$, where each $\Yb_j$ influences at most $k$ of the $\Xb_i$'s.
The following is essentially \cite[Theorem 1.1]{GLSS15}.

\begin{lemma}[Read-$k$ Chernoff]\label{lem:read_k_chernoff}
    Let $\Xb_1, \dots, \Xb_n$ be a read-$k$ family of indicator random variables with $\Pr[\Xb_i=1] \ge p$ for all $i$.
    Then for any $\delta \in (0,1)$, we have
    \[
        \Pr\sbra{\sum_i \Xb_i \le (1-\delta) pn} \le \exp\pbra{-\Omega\pbra{\frac{\delta^2 pn}{k}}}.
    \]
\end{lemma}

\section{Proof of \Cref{thm:main-lb}}\label{sec:main_pf}

In this section, we prove \Cref{thm:main-lb}, restated below for the reader's convenience.

\thmmainlb*

At a high level, the proof reduces to a setting in which many groups of output bits that are forced to be equal under $\Db_{n,\ell}$ do not share a common input cell.
For each such group, this absence of a shared input creates a dichotomy: either the sampler produces unequal bits with noticeable probability, immediately violating the defining constraints of $\Db_{n,\ell}$, or the entire group is biased toward one fixed constant value.
If the latter occurs for many groups, then the sampler places too much mass on outputs where many coordinates simultaneously take prescribed values, an event that is exponentially unlikely under $\Db_{n,\ell}$.
Thus in either case the sampler is far from the target distribution.
We now proceed to the details.

For clarity, we will write $\Ub$ for $\Ub_{[N]}$ and $\Db$ for $\Db_{n,\ell}$.
It will also be convenient to discuss the indices in the output of $f(\Ub^m)$ in a similar fashion to the definition of $\Db$ (in \Cref{def:distD}).
That is, we will abuse notation and often use (say) $\ib^{(1)}$ for the first $\log n$ bits of $f(\Ub^m)$.
It will be evident from context which distribution (i.e., $f(\Ub^m)$ or $\Db$) the random variables $\ib^{(j)}, \xb,\yb^{(j)}$ are being taken from.

Let $t > 0$ be a threshold parameter (eventually taken to be $\Theta(\ell d \log N)$), and let $S \subseteq [m]$ be the set of input symbols to the sampler $f$ that either affect more than $t$ output bits or affect one of the first $\ell \log n$ outputs (corresponding to $\ib^{(1)}, \dots, \ib^{(\ell)}$).
Observe that $s \coloneqq |S| \le  d\ell \log(n) + d(\ell+1)n/t$.
By conditioning on the values that the symbols in $S$ take, we can express $f(\Ub^m)$ as the mixture
\[
    f(\Ub^m) = \sum_{\rho \in [N]^S} \frac{1}{N^{s}}\cdot f(\Ub^{[m] \setminus S}, \rho).
\]

Below, we will write the distribution $f(\Ub^{[m] \setminus S}, \rho)$ more concisely as $\Fb_\rho$.
Observe that any such conditioning $\rho$ fixes the value of the shift vector $\ib(\rho) \coloneqq (\ib^{(1)}(\rho), \dots, \ib^{(\ell)}(\rho))$, where $\ib^{(j)}(\rho)$ is the value of $\ib^{(j)}$ after conditioning on $\rho$, and it guarantees that every unfixed input cell affects at most $t$ output bits.
We will use these properties to show that for certain restrictions $\rho$, the distribution $\Fb_\rho$ is far from $\Db$ conditioned on their shift vectors being equal.
Then, we amalgamate the behavior of these conditional distributions to obtain our final distance bound.

In order for $f(\Ub^m)$ to approximate $\Db$, there must be clusters of $k \coloneqq \ell+1$ output bits which all take the same value.
More precisely, for each shift vector $i \in (\{0,1\}^{\log n})^\ell$ and $u\in [n]$, define the equality block $B_u(i) = (\xb_u, \yb^{(1)}_{u - i^{(1)}}, \dots, \yb^{(\ell)}_{u - i^{(\ell)}})$.
Call such an $i$ \emph{bad} if at least $n/4$ of the equality blocks have some unfixed input cell $c \not\in S$ which affects every bit in the block (and \emph{good} otherwise); if a block has such an input symbol, we say the block is \emph{covered}.

\begin{claim}\label{clm:num_bad_i}
    The fraction of bad $i \in (\{0,1\}^{\log n})^\ell$ is at most $4d\left(\frac{t}{\ell n}\right)^\ell$.
\end{claim}
\begin{proof}
    For each unfixed input cell $c \not\in S$, define $A_{x,c} \subseteq [n]$ to be the set of $\xb$ coordinates affected by $c$.
    Similarly, for all $j\in [\ell]$ define $A_{j,c} \subseteq [n]$ to be the set of $\yb^{(j)}$ coordinates affected by $c$.
    Observe that $\sum_j |A_{j,c}| \le t$, and that a block $B_u(i)$ is covered by a cell $c$ only if $u \in A_{x, c}$ and $u - i^{(j)} \in A_{j,c}$ for all $j$.
    In a uniformly sampled $\ib$, each $\ib^{(j)}$ is chosen independently and uniformly at random from $\bin^{\log n}$, so the probability that $u - \ib^{(j)} \in A_{j,c}$ is $|A_{j,c}| / n$.
    Thus,
    \begin{align*}
        \E_{\ib \in (\{0,1\}^{\log n})^\ell}[\# \text{ blocks covered}] &\le \sum_{c \not\in S}\E_{\ib \in (\{0,1\}^{\log n})^\ell}[\# \text{ blocks covered by $c$}]  \\
        &= \sum_{c \not\in S} \sum_{u \in A_{x,c}} \prod_{j=1}^\ell \frac{|A_{j,c}|}{n} \tag{since $\ib^{(j)}$'s independent} \\
        &\le \sum_{c \not\in S} |A_{x,c}| \pbra{\frac{\sum_j |A_{j,c}|}{n\ell}}^\ell \tag{by AM-GM inequality} \\
        &\le \pbra{\sum_{c \not\in S} |A_{x,c}|} \cdot \pbra{\frac{t}{n\ell}}^\ell  \tag{by def'n of $S$} \\
        &\le dn \pbra{\frac{t}{n\ell}}^\ell. \tag{since $\xb$ affected by $\le dn$ inputs}
    \end{align*}
    Applying Markov's inequality concludes the proof.
\end{proof}

By \Cref{clm:num_bad_i}, we can primarily focus our attention on the case of good $i$'s, where we have many uncovered blocks.
We will show that the lack of correlation within the bits of such a block causes behavior which does not occur in the target distribution.
Below, we view $\Zb$ as the marginal distribution of $\Fb_\rho$ onto an uncovered block.
\begin{claim}\label{clm:Z_classification}
    Let $\Zb$ be the marginal distribution over some $k \ge 2$ output bits.
    If no input cell affects every bit in $\Zb$, then either
    \begin{enumerate}
        \item \label{itm:Znotallequal} $\Pr[\Zb \not\in \{0^k, 1^k\}] \ge 1/(100k)$, or
        \item \label{itm:Zfixed} $\Pr[\Zb = z] \ge 2/3$ for some fixed $z \in \{0^k, 1^k\}$.
    \end{enumerate}
\end{claim}
\begin{proof}
    Arbitrarily partition the input cells into sets $I_1, \dots, I_k$ where the cells in $I_j$ do not affect the $j$-th output bit of $\Zb$.
    We proceed by a hybrid argument.
    Consider two independent random inputs $\Ab, \Bb$.
    We define the sequence of inputs 
    \[
        \Ab = \Cb^{(0)}, \Cb^{(1)}, \dots, \Cb^{(k)} = \Bb,
    \]
    where $\Cb^{(j)}$ has the input symbols in $I_1 \cup \cdots \cup I_j$ taken from $\Bb$ and the rest from $\Ab$.
    We also define $\Zb^{(j)}$ to be the value of $\Zb$ on input $\Cb^{(j)}$, and observe that each $\Zb^{(j)}$ has the same marginal distribution as $\Zb$.
    Moreover, let $\delta = \Pr[\Zb \not\in \{0^k,1^k\}]$, $p_0 = \Pr[\Zb = 0^k]$, and $p_1 = \Pr[\Zb = 1^k]$.
    
    Whenever $\Zb^{(0)} = 0^k$ and $\Zb^{(k)} = 1^k$ (or vice-versa), there must exist some $\Zb^{(j)} \not\in \{0^k,1^k\}$.
    Indeed, any two consecutive outputs $\Zb^{(j-1)}, \Zb^{(j)}$ agree on their $j$-th bit, since the cells in $I_j$ do not affect it.
    Thus, we cannot transition directly from $0^k$ to $1^k$, so if the sequence contains both strings, then it must also contain some other output.
    This implies $2p_0p_1 \le (k+1)\delta$.
    Assume by contradiction that $\delta < 1/(100k)$ and $p_0, p_1 < 2/3$.
    Then $\min(p_0,p_1) > 1 - \frac{1}{100k} - \frac{2}{3} \ge \frac{1}{4}$, so
    \[
        \frac{1}{8} \le 2\left(\frac{1}{4}\right)^2 < \frac{k+1}{100k} \le \frac{1}{100} + \frac{1}{100k} \le \frac{1}{8} - \frac{11}{100},
    \]
    a contradiction.
\end{proof}

Fix a conditioning $\rho$ such that $\ib(\rho)$ is good.
That is, at least $3n/4$ blocks $B_u(\ib(\rho))$ are uncovered.
We break into cases depending on which conclusion of \Cref{clm:Z_classification} is satisfied by most of these blocks.
Below, let $n' = 3n/8$.

\paragraph{Case 1: Most blocks satisfy (\ref{itm:Znotallequal}).}

Let $B_1, \dots, B_{n'}$ be blocks satisfying (\ref{itm:Znotallequal}).
While they are not fully independent, one would expect their behavior to resemble that of independent random variables, since each unfixed input cell only affects a small number of blocks.\footnote{A greedy construction allows us to find $\Omega(n/dt)$ actually independent blocks, but pursuing this line of analysis gives worse bounds.}
More formally, we have that the indicator random variables $\{\indicator(B_j \not\in \{0^k,1^k\})\}_j$ form a read-$t$ family, so \Cref{lem:read_k_chernoff} implies
\[
    \Pr_{\Fb_\rho}\sbra{\sum_{j=1}^{n'} \indicator(B_j \not\in \{0^k,1^k\}) \ge \frac{n'}{200k}} \ge 1 - \exp\pbra{-\Omega\pbra{\frac{n}{kt}}}.
\]
For comparison, this event never occurs under $\Db$ conditioned on its shift vector being $\ib(\rho)$.

\paragraph{Case 2: Most blocks satisfy (\ref{itm:Zfixed}).}
Let $B_1, \dots, B_{n'}$ be blocks satisfying (\ref{itm:Zfixed}), where their most common values are $b_1, \dots, b_{n'}$, respectively.
Again applying \Cref{lem:read_k_chernoff}, we have
\[
    \Pr_{\Fb_\rho}\sbra{\sum_{j=1}^{n'} \indicator(B_j = b_j) \ge \frac{7n'}{12}} \ge 1 - \exp\pbra{-\Omega\pbra{\frac{n}{t}}}.
\]
For comparison, $\sum_{j=1}^{n'} \indicator(B_j = b_j)$ is distributed like the binomial $\Bin(n',1/2)$ under $\Db$ conditioned on its shift vector being $\ib(\rho)$, so this event occurs with probability at most $e^{-\Omega(n)}$ by \Cref{lem:chernoff}.
\newline

\noindent 
In either case, there exists an event $\Ecal_\rho$ which occurs with probability at least $1 - e^{-\Omega\pbra{\frac{n}{kt}}}$ under $\Fb_\rho$, but probability at most $e^{-\Omega(n)}$ under $\Db$ conditioned on its shift vector.
We now define the global event $\Ecal$ witnessing the TV distance between $f(\Ub^m)$ and $\Db$ to be 
\[
    \Ecal = \{\text{bad } \ib \} \cup \bigcup_{\rho \: : \: \ib(\rho) \text{ is good}} (\ib = \ib(\rho) \text{ and } \Ecal_\rho).
\]
Under $f(\Ub^m)$, each conditioning $\rho \in [N]^S$ produces a shift vector $\ib(\rho)$.
If $\ib(\rho)$ is bad, then the output automatically lies in $\Ecal$, and if $\ib(\rho)$ is good, then we have already shown it lies in $\Ecal$ with probability at least $1 - e^{-\Omega\pbra{\frac{n}{kt}}}$.
To analyze the probability under $\Db$, we apply a union bound to find that
\begin{align*}
    \Pr_\Db[\Ecal] &\le \Pr_\Db[\ib \text{ is bad}] + \sum_{\rho \: : \: \ib(\rho) \text{ is good}} \Pr_\Db\sbra{\Ecal_\rho \mid \ib = \ib(\rho)} \\
    &\le 4d\left(\frac{t}{\ell n}\right)^\ell + N^s \exp(-\Omega(n)). \tag{by \Cref{clm:num_bad_i}}
\end{align*}
Hence,
\[
    \tvdist{f(\Ub^m) - \Db} \ge 1 - \exp\pbra{-\Omega\pbra{\frac{n}{(\ell+1)t}}} - 4d\left(\frac{t}{\ell n}\right)^\ell - N^s \exp(-\Omega(n)).
\]
Setting $t = \Theta(\ell d \log N)$ with a sufficiently large implicit constant gives $s \le d\ell \log(n) + O(n/\log N)$, where $O(n/\log N)$ has a sufficiently small implicit constant, so we find the distance between $f(\Ub^m)$ and $\Db$ is at least
\[
    1 - \exp\pbra{-\Omega\pbra{\frac{n}{\ell^2 d \log N}}} - d\left(\frac{\Theta(d\log N)}{n}\right)^\ell - \exp\big(d\ell \log(n)\log(N) -\Omega(n)\big).
\]
This concludes the proof of \Cref{thm:main-lb}. 

\section{Acknowledgments}

We are grateful to anonymous reviewers for helpful feedback on an earlier version of this manuscript.

\bibliographystyle{alpha}
\bibliography{ref}

\end{document}